\documentclass[a4paper,USenglish,cleveref]{lipics-v2021}

\pdfoutput=1 
\hideLIPIcs  
\makeatletter
\let\@oddfoot\@empty
\makeatother

\usepackage{mathtools}

\theoremstyle{plain}
\newtheorem{fact}[theorem]{Fact}

\newcommand\card[1]{\left|#1\right|}

\DeclareMathOperator{\calF}{\mathcal{F}}

\DeclareMathOperator{\calI}{\mathcal{I}}

\DeclareMathOperator*{\E}{\mathbb{E}}

\DeclareMathOperator{\Tr}{Tr}

\DeclareMathOperator{\vc}{\mathbb{VC}}
\DeclareMathOperator{\un}{\mathbb{U}}
\DeclareMathOperator{\prj}{\mathbb{PR}}
\DeclareMathOperator{\af}{\mathbb{AF}}

\newcommand{\bbF}{\mathbb{F}}
\newcommand{\Q}{{\left\{0,1\right\}}}
\DeclareMathOperator{\IP}{IP}
\DeclareMathOperator{\sat}{sat}
\DeclareMathOperator{\tc}{tc}
\newcommand{\setbinom}[2]{{\left[ #1 \right]}^{#2}}

\newcommand\restr[2]{{
  \left.\kern-\nulldelimiterspace
  #1
  \right|_{#2}
}}

\nolinenumbers

\title{A Variant of the VC-dimension with Applications to Depth-3 Circuits}

\author{Peter Frankl}{R\'{e}nyi Institute, Budapest, Hungary}{peter.frankl@gmail.com}{}{}

\author{Svyatoslav Gryaznov}{Institute of Mathematics of the Czech Academy of Sciences, Prague, Czech Republic \and St.\ Petersburg Department of V.A. Steklov Institute of Mathematics of the Russian Academy of Sciences, Russia}{svyatoslav.i.gryaznov@gmail.com}{https://orcid.org/0000-0002-5648-8194}{Supported by GA{\v C}R grant 19-05497S. The Institute of Mathematics of the Czech Academy of Sciences is supported by RVO: 67985840.}

\author{Navid Talebanfard}{Institute of Mathematics of the Czech Academy of Sciences, Prague, Czech Republic}{talebanfard@math.cas.cz}{https://orcid.org/0000-0002-3524-9282}{Supported by GA{\v C}R grant 19-27871X.}

\authorrunning{P. Frankl, S. Gryaznov, and N. Talebanfard}

\Copyright{Peter Frankl, Svyatoslav Gryaznov, and Navid Talebanfard}



\acknowledgements{We are grateful to Vojt{\v e}ch R{\"o}dl for discussions and to ITCS'22 reviewers for useful comments.}

\EventEditors{Mark Braverman}
\EventNoEds{1}
\EventLongTitle{13th Innovations in Theoretical Computer Science Conference (ITCS 2022)}
\EventShortTitle{ITCS 2022}
\EventAcronym{ITCS}
\EventYear{2022}
\EventDate{January 31--February 3, 2022}
\EventLocation{Berkeley, CA, USA}
\EventLogo{}
\SeriesVolume{215}
\ArticleNo{19}

\begin{document}

\maketitle

\begin{abstract}
We introduce the following variant of the VC-dimension. Given $S \subseteq \Q^n$ and a positive integer $d$, we define $\un_d(S)$ to be the size of the largest subset $I \subseteq [n]$ such that the projection of $S$ on every subset of $I$ of size $d$ is the $d$-dimensional cube. We show that determining the largest cardinality of a set with a given $\un_d$ dimension is equivalent to a Tur{\'a}n-type problem related to the total number of cliques in a $d$-uniform hypergraph. This allows us to beat the Sauer--Shelah lemma for this notion of dimension. We use this to obtain several results on $\Sigma_3^k$-circuits, i.e., depth-$3$ circuits with top gate OR and bottom fan-in at most $k$:

\begin{itemize}
\item Tight relationship between the number of satisfying assignments of a $2$-CNF and the dimension of the largest projection accepted by it, thus improving Paturi, Saks, and Zane (Comput. Complex. '00).

\item Improved $\Sigma_3^3$-circuit lower bounds for affine dispersers for sublinear dimension. Moreover, we pose a purely hypergraph-theoretic conjecture under which we get further improvement.

\item We make progress towards settling the $\Sigma_3^2$ complexity of the inner product function and all degree-$2$ polynomials over $\bbF_2$ in general. The question of determining the $\Sigma_3^3$ complexity of $\IP$ was recently posed by Golovnev, Kulikov, and Williams (ITCS'21).
\end{itemize}
\end{abstract}

\section{Introduction}

Boolean circuits provide a natural model for computing Boolean functions. Given a Boolean function $f$ in variables $x_1, \ldots, x_n$, a circuit is a sequence $C = \langle g_1, \ldots, g_t\rangle$ of functions where each $g_i$ is either an input variable, its negation, or $g_i = g(g_{i_1}, g_{i_2})$ where $i_1, i_2 < i$ and $g$ is an arbitrary Boolean function. The output of the circuit on an input $x$ is given by the last function, that is, $C(x) = g_t(x)$. The size of the circuit $C$ is $t$, the length of the sequence. A well-known simple counting argument due to Shannon shows that almost all Boolean functions in $n$ variables require circuits of size $2^n/n$ (see~\cite{MR2895965} for a proof and more background on circuit complexity). Despite this fact, the best known circuit size lower bound for an explicit function is barely above $3n$~\cite{DBLP:conf/focs/FindGHK16}. It is known that the arguments used in this result, the so-called gate elimination technique, cannot even yield a lower bound of $5n$~\cite{DBLP:journals/jcss/GolovnevHKK18}. Thus, a super-linear size lower bound even for logarithmic depth circuits would make a breakthrough in the formidable wall of complexity theory.

\subsection{Valiant's program and depth-3 circuits}

Another natural model of computation is bounded-depth circuits with unbounded fan-in. Let $\Sigma_3^k$ denote the class of depth-$3$ circuits of the form $\mathrm{OR} \circ \mathrm{AND} \circ \mathrm{OR}$ with bottom fan-in at most $k$. Equivalently we can view  a $\Sigma_3^k$-circuit as an unbounded disjunction of $k$-CNF formulas. Valiant~\cite{MR0660702} formulated a program to prove super-linear lower bounds for $O(\log n)$-depth circuits. He showed that a lower bound of $2^{\omega(n/\log\log n)}$ for $\Sigma_3^{n^{\epsilon}}$ circuits for some $\epsilon > 0$ implies a super-linear lower bound for $O(\log n)$-depth fan-in 2 circuits. Furthermore, he showed that if we restrict the depth-3 circuits to $\Sigma_3^{O(1)}$ then a lower bound larger than $2^{n/2}$ implies a super-linear lower bound for series-parallel circuits. This gives a strong motivation to prove $\Sigma_3^k$-circuit lower bounds for a fixed function for every constant $k$. In this direction Paturi, Pudl{\'a}k, and Zane~\cite{DBLP:conf/focs/PaturiPZ97} proved a lower bound of $\Omega(2^{n/k})$ for the parity function. Later Paturi, Pudl{\'a}k, Saks, and Zane~\cite{DBLP:journals/jacm/PaturiPSZ05} using similar but stronger techniques gave a lower bound of $\Omega(2^{n\pi^2/6k})$ for the characteristic function of the BCH code. This remains the best known result of this type.

Recently Golovnev, Kulikov, and Williams~\cite{DBLP:conf/innovations/GolovnevKW21} in an insightful revisiting of Valiant's program showed among other things that a lower bound of $2^{n - o(n)}$ for $\Sigma_3^{16}$-circuits implies a $3.9n - o(n)$ size lower bound for unrestricted circuits. This is a significant result since as we mentioned earlier the best lower bounds for unrestricted circuits are much weaker. This gives a strong motivation to study $\Sigma_3^k$-circuits for small values of $k$, and in fact, a result in this direction is known. Paturi, Saks, and Zane~\cite{DBLP:journals/cc/PaturiSZ00} proved a $2^{n - o(n)}$ lower bound for $k = 2$. For $k \ge 3$ no such bounds are known.

\subparagraph*{Lower bound arguments.} Let us briefly recall the general strategy in $\Sigma_3^k$-circuit lower bounds. Let $f$ be our hard function. If we can show that any $k$-CNF formula $F$ which is consistent with $f$, that is $F(x) \le f(x)$ for all $x$, has at most $R$ satisfying assignments, then it follows that $f$ requires $\Sigma_3^k$-circuits of size at least $\card{f^{-1}(1)}/R$. This is because a $\Sigma_3^k$-circuit computing $f$ gives a covering of $f^{-1}(1)$ using the sets of satisfying assignments of $k$-CNF formulas which are consistent with $f$.

The specific execution of this argument for $\Sigma_3^2$ is as follows. Given a CNF formula $\phi$, we denote by $\sat(\phi)$ the set of satisfying assignments of $\phi$.~\cite{DBLP:journals/cc/PaturiSZ00} showed that if $S = \sat(\phi)$ for a $2$-CNF formula $\phi$ in $n$ variables and $\card{S} = 2^{\Omega(n)}$, then $S$ contains a \textit{projection} of dimension $\Omega(n)$. A projection is simply an affine space defined by equations of the form $x = 0$, $x = 1$, $x = y$ or $x = 1 - y$. Thus, if we have a function $f$ such that $f^{-1}(1)$ has size at least $2^{n - o(n)}$ and does not contain any projection of linear dimension, then $f$ requires $\Sigma_3^2$-circuits of size $2^{n - o(n)}$. It turns out that explicit constructions of even more general such functions exist. These are called \textit{affine dispersers} for sublinear dimension, functions which are not constant under any affine space of some $o(n) $ dimension (see, e.g.,~\cite{DBLP:journals/siamcomp/Ben-SassonK12})\footnote{Note that these functions were explicitly constructed more  than a decade after~\cite{DBLP:journals/cc/PaturiSZ00} appeared. That paper got around this by constructing a disperser from a pseudo-random distribution.}.
One may ask if it is possible to extend the above result regarding projections to $k$-CNFs for $k \ge 3$ and thus obtain $\Sigma_3^k$ lower bounds. However,~\cite{DBLP:journals/cc/PaturiSZ00} showed that there are $4$-CNFs with exponentially many satisfying assignments which have only projections of constant dimension (later we will show that there are even $3$-CNFs with this property). This limits the applicability of projections to $k$-CNFs for $k \ge 3$, but it is conceivable that every sufficiently large set of satisfying assignments of a $k$-CNF contains a large \emph{general} affine subspace.


\subparagraph*{Affine dispersers.} 
Given $S \subseteq \Q^n$ let $\af(S)$ denote the dimension of the largest affine space contained in $S$.
Given $c$ define
\[
  \operatorname{af}_k(c) \coloneqq \inf_{n} \min \left\{ \af(\sat(\phi))/n : \text{$\phi$ is $k$-CNF in $n$ variables, $\card{\sat(\phi)} \ge 2^{cn}$} \right\}.
\]
Note that $\operatorname{af}_k(0) = 0$ and $\operatorname{af}_k(1) = 1$.
For every $k$ define $c_k$ to be the infimum of $c$ for which $\operatorname{af}_k(c) > 0$. Observe that an affine disperser for sublinear dimension requires $\Sigma_3^k$-circuits of size $2^{(1 - c_k)n - o(n)}$ (we may assume that our function has $2^{n - o(n)}$ ones, otherwise the negation of the function, which is also an affine disperser, does). In particular, if $c_k = 0$ for every $k$ then we get superlinear lower bounds for series-parallel circuits, and if $c_{16} = 0$ then affine dispersers require general circuits of size $3.9n - o(n)$. Interestingly the state-of-the-art general circuit lower bounds are proved for such affine dispersers~\cite{DBLP:conf/focs/FindGHK16}. Thus finding upper bounds on $c_k$ is a justified direction to explore. So far we know only that $c_2 = 0$~\cite{DBLP:journals/cc/PaturiSZ00}.
For arbitrary $k$ the best upper bound to our knowledge can be easily inferred from the Switching Lemma (see~\cite{DBLP:journals/tcs/MiltersenRW05, DBLP:conf/coco/Rossman19}):
any $k$-CNF in $n$ variables has a decision tree representation of size $2^{(1 - 1/Ck)n}$, where $C > 1$ is a universal constant, and thus a $k$-CNF that accepts significantly more than $2^{(1 - 1/Ck)n}$ assignments, in particular, accepts a large subcube, which is the simplest form of an affine space. The constant $C$ comes from the constant appearing in the Switching Lemma which can be set to 10. It follows that $c_k \le 1 - \frac{1}{10k}$.

\subsection{Our contributions}

We introduce a variant of the VC-dimension which allows better Sauer--Shelah type lemmas~\cite{DBLP:journals/jct/Sauer72, MR307903}. Recall that for a set $S \subseteq \Q^n$ the VC-dimension of $S$, $\vc(S)$, is defined to be the size of the largest subset $I \subseteq [n]$ such that $S$ projected on coordinates in $I$ is the $\card{I}$-dimensional cube. This is a fundamental concept from learning theory~\cite{vc-ucrfep-71} which is also extensively studied in combinatorics (see, e.g.,~\cite{MR319773, DBLP:journals/jct/BollobasR95, DBLP:journals/jlms/AlonMS19}). It is also used in circuit complexity (see, e.g.,~\cite{DBLP:journals/cc/PaturiSZ00, DBLP:journals/jcss/ImpagliazzoPZ01} for depth-3, and~\cite{DBLP:journals/cc/EdmondsIRS01, DBLP:journals/cc/Meir20} for general circuits). Applications of the VC-dimension usually go through the Sauer--Shelah lemma which states that if $\card{S} > \sum_{i = 0}^r\binom{n}{i}$ then $\vc(S) \ge r + 1$. This bound is tight and it is sufficient for most applications since it implies that if $\card{S} \ge 2^{\Omega(n)}$ then $\vc(S) \ge \Omega(n)$. However, this bound cannot guarantee the VC-dimension to be bigger than $n/2$ for sets of size $2^{\Omega(n)}$. To see this consider the set of all $n$-bit strings with Hamming weight at most $n/2$. This set has size $2^{n - 1}$ but VC-dimension only $n/2$.

\subparagraph*{A variant of the VC-dimension.}
The variant we consider is very natural.
Given a set $S \subseteq \Q^n$ and a positive integer $d$, $\un_d(S)$ is the size of the largest subset $I \subseteq [n]$ such that the projection of $S$ to every subset of $I$ of size $d$ is the $d$-dimensional cube. We show that the size of the largest set $S \subseteq \Q^n$ with $\un_d(S) = r$ is the same as the maximum number of cliques in an $n$-vertex $d$-uniform hypergraph with no clique of size $r + 1$. Luckily for $d = 2$ this quantity can be computed exactly from a generalization of Tur{\'a}n's theorem due to Zykov, and it turns out to be ${\left(\frac{n}{r} + 1\right)}^r$. Note that this immediately overcomes the $n/2$ barrier of the VC-dimension mentioned before: for the above example we have $\un_{n/2}$ dimension exactly $n$ and in general for every $\epsilon > 0$ there exists $\delta < 1$ such that if $\card{S} > 2^{\delta n}$ then $\un_2(S) \ge (1 - \epsilon)n$. For larger values of $d$ we can determine this bound when $r \ge (1 - 1/d)n$. For other values, we state a conjecture that extends Zykov's theorem to $d$-uniform hypergraphs.

\subparagraph*{Applications.} We obtain several results regarding depth-$3$ circuits.

\begin{itemize}
\item Bottom fan-in $2$: The first application is a tightening of~\cite{DBLP:journals/cc/PaturiSZ00} relating the dimension of the largest projection contained in the set of satisfying assignments of a $2$-CNF and its size. This allows us to obtain the following results.

\begin{itemize}
\item \textit{Lower bounds for weaker affine dispersers:} We prove lower bounds on the size of $\Sigma_3^2$ circuits for affine dispersers for linear dimension. This is interesting since~\cite{DBLP:journals/cc/PaturiSZ00} does not give anything for affine dispersers for dimension bigger than $n/2$.

\item \textit{Progress on the complexity of the inner product function (IP):}
The general strategy for proving $\Sigma_3^k$ circuit lower bounds described above does not give optimal bounds for some functions, notably the inner product function $\IP$. \cite{DBLP:conf/innovations/GolovnevKW21} also poses the question of proving tight bottom fan-in $3$ lower bounds for $\IP$. But tight lower bounds are not known even for bottom fan-in $2$ and here we focus on this case. We show that any $2$-CNF consistent with the $\IP$ on $n$ variables accepts at most $3^{n/2}$ assignments and this is tight. Thus, we obtain a $\Sigma_3^2$-circuit size lower bound of $2^{0.20n}$ that is worse than the best known $2^{0.25n}$ lower bound, which follows from a reduction to parity. However, we show that there is a unique $2$-CNF consistent with IP achieving this bound. This suggests that an alternative approach to lower bounds, namely the stability approach, might be fruitful. Stability results show that a large set avoiding a certain forbidden structure looks very similar to the unique extremal set (see, e.g.,~\cite{MR2419936,DBLP:journals/jct/Furedi15}). In circuit complexity we are only aware of one such result, which can be found in a work of Dinur and Meir~\cite{DBLP:journals/cc/DinurM18}. Perhaps it is possible to show that a $2$-CNF consistent with $\IP$ which has many satisfying assignments has a particular structure, and this might allow us to prove the desired lower bound.

\item \textit{Complexity of degree-2 polynomials over $\bbF_2$:} We show that any such polynomial in $n$ variables requires $\Sigma_3^2$-circuits of size $2^{n/10}$. Impagliazzo, Paturi, and Zane~\cite{DBLP:journals/jcss/ImpagliazzoPZ01} showed that almost all degree-2 polynomials require $\Sigma_3^k$-circuits of size $2^{n - o(n)}$ for $k = O(1)$. Thus developing lower bound arguments for these functions contributes to the program of finding explicit hard degree-2 polynomials. The complexity of these functions has been studied previously for depth-3 circuits with XOR at bottom by Cohen and Shinkar~\cite{DBLP:conf/innovations/CohenS16}.
\end{itemize}

\item Bottom fan-in 3: Assuming that a $3$-CNF has sufficiently many satisfying assignments we give a large projection contained in the set of satisfying assignments which also yields a $\Sigma_3^3$ lower bound for affine dispersers. This follows from our lower bound on $\un_3$ for sufficiently large sets. In particular, it implies that $c_3 \le \frac{\log 7}{3} \simeq 0.936$. Note that this is less than the $\frac{29}{30} \simeq 0.966$ bound which follows from the Switching Lemma. Although this improvement is modest, the underlying conceptual arguments seem to provide new insight. Our technique poses a Tur{\'a}n-type conjecture for hypergraphs which, if true, would imply $c_3 \le 0.707$.

\end{itemize}

\section{The \texorpdfstring{$\un_d$}{U\_d} dimension}

\begin{definition}
Let $\calF \subseteq 2^{[n]}$ be a set system and let $I \subseteq [n]$. The \emph{trace} of $\calF$ on $I$ is defined by $\Tr_{\calF}(I) \coloneqq \{A \cap I : A \in \calF\}$. Equivalently, viewing $\calF$ as a subset of $\Q^n$, $\Tr_{\calF}(I)$ is the set of distinct vectors obtained by projecting $\calF$ on the coordinates in $I$.
\end{definition}

\begin{definition}
Let $\calF \subseteq 2^{[n]}$ be a set system. We say that $I \subseteq [n]$ is \emph{shattered} if $\card{\Tr_{\calF}(I)} = 2^{\card{I}}$. Given $\calF$  the VC-dimension of $\calF$, denoted by $\vc(\calF)$, is the size of the largest shattered set.
\end{definition}

\begin{definition}[$d$-Universality]
  Let $\calF \subseteq 2^{[n]}$ be a set system and $d$ a positive integer. We say that $I \subseteq [n]$ is \emph{$d$-universal} for $\calF$ if $\card{I} \ge d$ and every $J \subseteq I$ with $\card{J} = d$ is shattered. We say that $\calF$ has \emph{property $U(r, d)$} if there exists $I \subseteq [n]$ of size $r$ which is $d$-universal. We denote by $u(n, r, d)$ the cardinality of the largest system of subsets of $[n]$ which does not have property $U(r + 1, d)$. We write $\un_d(\calF)$ to denote the size of the largest $d$-universal set for $\calF$.
\end{definition}

It immediately follows from the definition that if $\vc(\calF) \ge d$ then $\un_d(\calF) \ge \vc(\calF)$. To prove an upper bound on $u(n, r, d)$ we observe that it is sufficient to consider downward closed systems. We adopt the squashing argument of Frankl~\cite{DBLP:journals/jct/Frankl83}.

\begin{lemma}\label{lm:downclosed}
  Let $\calF \subseteq 2^{[n]}$ be a set system not having property $U(r+1, d)$ such that $\sum_{A \in \calF} \card{A}$ is minimized over all such families of cardinality $\card{\calF}$. Then $\calF$ is a downward closed family.
\end{lemma}

\begin{proof}
Assume for a contradiction that $\calF$ is not downward closed. Then there exists $A \in \calF$ and $i \in [n]$ such that $A \setminus \{i\} \not \in \calF$. For any $B \subseteq [n]$ we define
\begin{equation*}
B' \coloneqq \begin{cases}
        B \setminus \{i\}  & \text{if $i \in B$ and $B \setminus \{i\} \not \in \calF$}\\
        B & \text{otherwise.}
\end{cases}
\end{equation*}
We now define $\calF' \coloneqq \{B' : B \in \calF\}$. Note that $\card{\calF'} = \card{\calF}$ and since $A' = A \setminus \{i\}$, $\sum_{C
\in \calF'}\card{C} < \sum_{D \in \calF}\card{D}$. Therefore, by the minimality assumption, $\calF'$ has property $U(r + 1, d)$ and hence there exists $I \subseteq [n]$ with $\card{I} = r + 1$, which is $d$-universal for $\calF'$. We will show that $I$ is $d$-universal also for $\calF$, which is a contradiction. Since $\calF$ and $\calF'$ agree on all elements except for $i$, we may assume that $i \in I$. By the same reasoning $I \setminus \{i\}$ is $d$-universal for $\calF$. Therefore, it remains to show that for any $J \subseteq I \setminus \{i\}$ with $\card{J} = d-1$, $\card{\Tr_{\calF}(J \cup \{i\})} = 2^d$. We will show that for any $S \subseteq J$, we have that both $S$ and $S \cup \{i\}$ are in $\Tr_{\calF}(J \cup  \{i\})$. By $d$-universality $S \cup \{i\} \in \Tr_{\calF'}(J \cup \{i\})$ and hence there exists $E \in \calF'$ such that $S \cup \{i\} = E \cap (J \cup \{i\})$. Since $i\in E$, by construction of $\calF'$, it follows that $E \in \calF$ and hence $S \cup \{i\} \in \Tr_{\calF}(J \cup \{i\})$. Furthermore, again since $i \in E$ and by construction of $\calF'$, $E \setminus \{i\} \in \calF$. Since $S = (E \setminus \{i\}) \cap (J \cup \{i\})$, we have $S \in \Tr_{\calF}(J \cup \{i\})$.
\end{proof}

Given a $d$-uniform hypergraph (or a $d$-graph) $H = (V, E)$, a clique $S \subseteq V$ is a subset of vertices such that either $\card{S} < d$ or if $\card{S} \ge d$ then any subset of $S$ of size $d$ is a hyperedge in $E$. Analogously, $S$ is an independent set if it does not contain any hyperedge. We denote the $d$-uniform clique of size $t$ by $K^d_t$. Let us denote by $k(n, r, d)$ the maximum number of cliques in a $K^d_{r + 1}$-free $d$-graph on $n$ vertices.

\begin{lemma}\label{thm:ud}
For every $n \ge r \ge d$, $u(n, r, d) = k(n, r, d)$.
\end{lemma}

\begin{proof}
To show the lower bound, let $H = ([n], E)$ be a $K^d_{r+1}$-free $d$-graph achieving the maximum number of cliques. We define
\[
  \calF \coloneqq \{S \subseteq [n] : \text{$S$ is a clique in $H$}\}.
\]

Note that by construction $\calF$ is downward closed. Assume for a contradiction that there exists $I \subseteq [n]$ of size $r + 1$ which is $d$-universal for $\calF$. By $d$-universality and downward closedness, every subset of $I$ of size $d$ is in $\calF$ which implies that $I$ is a clique in $H$.

In the other direction let $\calF$ be a system of maximum size not having property $U(r + 1, d)$. By \cref{lm:downclosed} we may assume that $\calF$ is downward closed. We define a $d$-graph $H = ([n], E)$ as follows:
\[
  E \coloneqq \{S \in \calF: \card{S} = d\}.
\]

Since $\calF$ is downward closed, any clique $S \subseteq [n]$ is $d$-universal for $\calF$. Therefore, $H$ is $K^d_{r + 1}$-free. Note furthermore that each $S \in \calF$ gives a clique in $H$. Thus, the size of $\calF$ is bounded by the total number of cliques in $H$.
\end{proof}

Using \cref{thm:ud} and a generalization of Tur{\'a}n's Theorem, which has been rediscovered many times, we can determine $u(n, r, 2)$ precisely. Recall that the \emph{Tur{\'a}n graph} $T_{n, r}$ is the complete $n$-vertex $r$-partite graph with parts of sizes as equal as possible. 

\begin{theorem}[Zykov~\cite{zykov1949some}, Sauer~\cite{SAUER1971109}, Alekseev~\cite{Alekseev}]\label{thm:zyk}
Let $G$ be a $K_{r + 1}$-free graph on $n$ vertices. Then $k(G) \le k(T_{n, r}) \le {(\frac{n}{r}+1)}^r$.
\end{theorem}

Applying \cref{thm:ud} and \cref{thm:zyk} immediately implies the following.

\begin{theorem}\label{thm:utwo}
For every $n \ge r$, $u(n, r, 2) = k(T_{n, r}) \le {({\frac{n}{r}} + 1)}^r$. It follows that for every $\calF \subseteq 2^{[n]}$, $\card{\calF} \le {(\frac{n}{\un_2(\calF)} + 1)}^{\un_2(\calF)}$.
\end{theorem}

We now determine $u(n, r, d)$ when $r$ is sufficiently large. Note that by complementation $k(n, r, d)$ is the same as the maximum number of independent sets in an $n$-vertex $d$-graph with no independent set of size $r + 1$. Given a hypergraph $H = (V, E)$, a transversal $T \subseteq V$ is a subset of vertices such that every edge of $H$ contains at least one vertex from $T$. Denote by $i(n, r, d)$ the maximum number of independent sets in an $n$-vertex $d$-graph with no transversal of size $r - 1$. It is easy to see that
\begin{equation}\label{eqn:trans}
i(n, r, d) = k(n, n - r, d)
\end{equation}
since an $n$-vertex $d$-graph has no transversal of size $r - 1$ if and only if it does not have any independent set of size $n - r + 1$.

\begin{theorem}\label{thm:large-r}
  Let $r \le n/d$ and let $H = (V, E)$ be a $d$-graph on $n$ vertices with no transversal of size $r - 1$ and maximum possible number of independent sets. Then $H$ is the disjoint union of $r$ hyperedges and $n-rd$ isolated vertices. Consequently, $i(n, r, d) = 2^{n - rd}{(2^d - 1)}^r$.
\end{theorem}

To prove this theorem we need the following auxiliary lemma.

\begin{lemma}\label{thm:mu_bound}
  Let $X$ be a set of size $d$. Consider a distribution $\mu$ on the subsets of $X$ with the following properties:
  \begin{enumerate}
    \item $\mu(X) = 0$.
    \item $\mu(F) \ge \mu(F')$ if $F \subseteq F'$.
  \end{enumerate}

  Then
  \[
    \E_{F \sim \mu} \left[\card{\overline{F}}\right] \ge d \frac{2^{d-1}}{2^d-1}.
  \]

  The equality holds if and only if $\mu(F) = \frac{1}{2^d-1}$ for every $F \subsetneq X$. In other words, $\E\limits_{F \sim \mu} \left[\card{\overline{F}}\right]$ is minimized if $\mu$ is the uniform distribution over all non-full sets.
\end{lemma}
\begin{proof}
  Let us denote by ${[X]}^i$ the set of subsets of $X$ of size $i$.
  Define $\nu$ on $\{0,1,\ldots,d\}$ as follows:
  \[
    \nu(i) \coloneqq \frac{\sum\limits_{F \in \setbinom{X}{i}} \mu(F)}{\binom{d}{i}}.
  \]

  \begin{claim}\label{lm:nu_decresing}
    For every $i$, $\nu(i) \ge \nu(i+1)$.
    The equality holds if and only if $\mu(F) = \mu(F')$ for every $F \subseteq F' \subsetneq X$, where $\card{F}=i$ and $\card{F'}=i+1$.
  \end{claim}
  \begin{claimproof}
    We need to show that
    \[
      \frac{\sum\limits_{F \in \setbinom{X}{i}} \mu(F)}{\binom{d}{i}} \ge \frac{\sum\limits_{F \in \setbinom{X}{i+1}} \mu(F)}{\binom{d}{i+1}},
    \]
    which is equivalent to
    \begin{equation}\label{eq:double_sum}
      (d-i) \sum\limits_{F \in \setbinom{X}{i}} \mu(F) \ge (i+1) \sum\limits_{F \in \setbinom{X}{i+1}} \mu(F).
    \end{equation}

    Consider the sum
    \[
      \sum\limits_{\substack{(F, F'), \text{ where } F \subseteq F', \\ F \in \setbinom{X}{i}, F' \in \setbinom{X}{i+1}}} \mu(F).
    \]

    Each $\mu(F)$, where $\card{F} = i$, appears exactly $d-i$ times since there are $d-i$ choices of $F'$ such that $F \subseteq F'$ and $\card{F'}=i+1$. Thus, it is equal to
    \begin{equation}\label{eq:s_upper}
      \sum\limits_{F \in \setbinom{X}{i}} (d-i) \mu(F).
    \end{equation}

    Similarly,
    \begin{equation}\label{eq:s_lower}
      \sum\limits_{\substack{(F, F'), \text{ where } F \subseteq F', \\ F \in \setbinom{X}{i}, F' \in \setbinom{X}{i+1}}} \mu(F') = \sum\limits_{F' \in \setbinom{X}{i+1}} (i+1) \mu(F').
    \end{equation}

    By the second property of $\mu$,~\eqref{eq:s_upper} is at least~\eqref{eq:s_lower}, which gives us~\eqref{eq:double_sum}.

    For the second part note that equality in~\eqref{eq:double_sum} holds if and only if $\mu(F) = \mu(F')$, where $F \subseteq F'$, $F \in \setbinom{X}{i}$, and $F' \in \setbinom{X}{i+1}$.
  \end{claimproof}

  By definition of $\nu$,
  \[
    \E_{F \sim \mu} \left[\card{\overline{F}}\right] = \sum\limits_{F \subseteq X} (d - \card{F}) \mu(F) = \sum\limits_{i=0}^d (d-i) \binom{d}{i} \nu(i).
  \]

  We can rewrite the last sum as
  \begin{equation}\label{eq:nu_expansion}
    \sum\limits_{j=0}^{d-1} \sum\limits_{i=0}^j \binom{d}{i} \nu(i).
  \end{equation}

  We need the following simple fact.
  \begin{lemma}\label{fact:sequences}
    Let ${\{a_i\}}_{i=1}^n$ be a sequence of non-decreasing numbers $a_1 \ge a_2 \ge \ldots \ge a_n$. Let ${\{b_i\}}_{i=1}^n$ be a sequence of non-negative numbers with $\sum\limits_{i=1}^n b_i > 0$. Then for every $1 \le m \le n$
    \[
      \sum\limits_{i=1}^m a_i b_i \ge \frac{\sum\limits_{i=1}^m b_i}{\sum\limits_{i=1}^n b_i} \sum\limits_{i=1}^n a_i b_i.
    \]
  \end{lemma}

  \begin{proof}
  	Define a random variable $X$ on $[n]$ which takes value $i$ with probability 
    \[\frac{b_i}{\sum\limits_{j = 1}^n b_j}.\]
    Observe that \[\E[a_X] = \frac{1}{\sum\limits_{i = 1}^n b_i}\sum\limits_{i = 1}^n a_i b_i,\]
    and \[\E[a_X|X \le m] = \frac{1}{\sum\limits_{i = 1}^m b_i}\sum\limits_{i = 1}^m a_i b_i.\] Using a simple coupling argument we show that $\E\left[a_X \;\middle|\; X \le m\right] \ge \E[a_X]$ which gives the result. We jointly sample $(A, B)$ such that $A$ is distributed as $X$ and $B$ is distributed as $X$ conditioned on $X \le m$. Furthermore, we guarantee that $B \ge A$ which by the assumption that $a_1 \ge \ldots \ge a_n$ implies $a_B \ge a_A$.

  	We first sample $A$. If $A \le m$ then we set $B = A$. Otherwise, we sample $B$ as $X$ conditioned on $X \le m$. It is easy to see that $(A, B)$ satisfies our requirements.
  \end{proof}

  Since $\mu$ is a distribution and $\mu(X) = 0$, we have
  \begin{equation}\label{eq:nu_total_sum}
    \sum\limits_{i=0}^{d-1} \binom{d}{i} \nu(i) = \sum\limits_{F \subsetneq X} \mu(F) = 1.
  \end{equation}

  From \cref{lm:nu_decresing}, \cref{fact:sequences} (for $a_i = \nu(i)$ and $b_i = \binom{d}{i}$), and~\eqref{eq:nu_total_sum} it follows that for every $0 \le j \le d-1$
  \begin{equation}\label{eq:nu_tail_small}
    \sum\limits_{i=0}^j \binom{d}{i} \nu(i) \ge \frac{\sum\limits_{i=0}^j \binom{d}{i}}{2^d-1}.
  \end{equation}

  Hence, we have the following lower bound on~\eqref{eq:nu_expansion}:
  \begin{equation}\label{eq:mu_bound}
    \sum\limits_{j=0}^{d-1} \sum\limits_{i=0}^j \binom{d}{i} \nu(i) \ge \frac{\sum\limits_{j=0}^{d-1} \sum\limits_{i=0}^j \binom{d}{i}}{2^d-1} = \frac{\sum\limits_{i=0}^d (d-i) \binom{d}{i}}{2^d-1} = d \frac{2^{d-1}}{2^d-1}.
  \end{equation}

  Observe that~\eqref{eq:mu_bound} is an equality if and only if for every $0 \le j \le d-1$~\eqref{eq:nu_tail_small} is an equality.
  Equivalently, for every $i$ we have $\nu(i) = \frac{1}{2^d-1}$.
  It follows from the second part of \cref{lm:nu_decresing} that $\mu(F) = \mu(F')$ for every $F \subseteq F' \subsetneq e$. In particular, for every $F \subsetneq e$ we have $\mu(F) = \mu(\varnothing) = \nu(0) = \frac{1}{2^d-1}$.
\end{proof}

This lemma can be used to prove the following result about the number of independent sets in a $d$-graph.
Let us denote by $i(H)$ the number of independent sets in a hypergraph $H$.

\begin{lemma}\label{lm:indsets_vertex_removal}
  Let $e = \{u_1,\ldots,u_d\}$ be an edge of a $d$-graph $H$. Then there exists $u \in e$ such that
  \[
    i(H) \le \frac{2^d-1}{2^{d-1}} i(H \setminus u).
  \]

  The equality holds if and only if every no other edge in $H$ intersects $e$.
\end{lemma}
\begin{proof}
  We partition the independent sets in $H$ by their ``footprint'' on $e$:
  \begin{equation*}
    \calI_F \coloneqq \left\{I : I \text{ is an independent set of } H, I \cap e = F\right\}, \text{ where } F \subseteq e.
  \end{equation*}

  Since removing any subset of vertices from an independent set leaves it independent, $\card{\calI_{F}} \ge \card{\calI_{F'}}$ if $F \subseteq F'$. Also, $\calI_e = \varnothing$ since $e$ is an edge.

   For $u \in e$, we can express the number of independent sets in the hypergraph $H \setminus u$ in terms of $\calI_F$.
  \begin{equation*}
    i(H \setminus u) = \sum\limits_{F \subseteq e \setminus \{u\}} \card{\calI_F}.
  \end{equation*}

  Thus, we have the following:
  \begin{equation}\label{eq:sum_H_minus_u}
    \sum\limits_{u \in e} i(H \setminus u) = \sum\limits_{F \subseteq e} (d - \card{F}) \card{\calI_F}.
  \end{equation}

  Consider a distribution $\mu$ defined on the subsets of $e$ as follows:
  \[
    \mu(F) \coloneqq \frac{\card{\calI_F}}{i(H)}.
  \]

  Clearly, $\mu$ satisfies all the conditions of \cref{thm:mu_bound}. Hence,
  \begin{equation}\label{eq:mu_thm_applied_to_indsets}
    \sum\limits_{F \subseteq e} (d-\card{F}) \card{\calI_F} \ge d \frac{2^{d-1}}{2^d-1} i(H).
  \end{equation}

  Applying~\eqref{eq:mu_thm_applied_to_indsets} to~\eqref{eq:sum_H_minus_u} gives
  \begin{equation*}\label{eq:indsets_vertex_removal}
    d \frac{2^{d-1}}{2^d-1} i(H) \le \sum\limits_{u \in e} i(H \setminus u) \le d \max\limits_{u \in e} i(H \setminus u).
  \end{equation*}

  This concludes the proof of the first part of the statement.

  For the second part \cref{thm:mu_bound} also implies that~\eqref{eq:mu_thm_applied_to_indsets} is an equality if and only if $\mu(F) = \frac{1}{2^d-1}$ for every $F \subsetneq e$.
  Consequently,
  \begin{equation}\label{eq:I_F_are_the_same}
    \card{\calI_F} = \card{\calI_\varnothing}.
  \end{equation}

  For every $F \subsetneq e$, consider an injective function $b_F\colon \calI_F \to \calI_\varnothing$ defined as follows:
  \[
    b_F(I) \coloneqq I \setminus F.
  \]

  It follows from~\eqref{eq:I_F_are_the_same} that $b_F$ is a bijection.

  Assume that there exists another edge $e'$ such that $F = e \cap e' \neq \varnothing$. $I = e' \setminus F$ is an independent set (its size is smaller than $d$), and, since $b_F$ is a bijection, $I \cup F = e'$ must be an independent set, which is a contradiction.
\end{proof}

Now we can finally prove \cref{thm:large-r}.
\begin{proof}[Proof of \cref{thm:large-r}]
  We prove it by induction on $n$ and $r$. For the case $r=1$, $H$ must be non-empty. Since removing an edge increases the number of independent sets, we can remove all but one edges from $H$. The hypergraph with exactly one edge has $2^{n - d}(2^d - 1)$ independent sets.

  For the inductive step, we use the bound from \cref{lm:indsets_vertex_removal}.
  Let us denote by $\tau(H)$ the size of a transversal of minimum size in $H$.
  Let $u$ be a vertex of $H$ such that $u$ is contained in at least one edge of $H$. If after removing $u$ the transversal number does not drop, we can remove every edge incident to $u$, and the resulting graph would not have a transversal of size $r-1$, but would have more independent sets than $H$. Thus, without loss of generality, we can assume that for every non-isolated vertex $u$, $\tau(H \setminus u) = \tau(H) - 1$.

  Clearly, $H$ consists of at least one edge. Let $e$ be an edge of $H$. \cref{lm:indsets_vertex_removal} together with the induction hypothesis imply that
  \[
    i(H) \le \frac{2^d-1}{2^{d-1}} i(H \setminus u) \le \frac{2^d-1}{2^{d-1}} 2^{(n-1)-(r-1)d} {(2^d-1)}^{r-1} = 2^{n-rd}{(2^d-1)}^r,
  \]
  and we have an equality here only if $e$ does not intersect any other edge in $H$.
\end{proof}

The next theorem follows immediately from \cref{thm:large-r} and~\eqref{eqn:trans}.

\begin{theorem}\label{thm:large-r-num}
  Let $r \ge (1 - \frac{1}{d})n$. Then $u(n, r, d) = 2^{n - (n - r)d}{(2^d - 1)}^{n - r}$.
\end{theorem}

In our applications we only use the upper bound on $u$. We conjecture that the natural extension of \cref{thm:zyk} to $d$-graphs holds. Recall the definition of binomial coefficients to real numbers. Given a positive real $x$ and an integer $k$ with $x \ge k$ we define $\binom{x}{k} \coloneqq \frac{x(x-1)\ldots(x -k + 1)}{k!}$. Furthermore, we define $V(x, d) \coloneqq \binom{x}{0}+\binom{x}{1}+\cdots+\binom{x}{d}$. In particular if $x$ is a positive integer, $V(x, d)$ is the size of the Hamming ball of radius $d$ in the $x$-dimensional cube.

\begin{conjecture}\label{conj:turan}
  Let $H$ be an $n$-vertex $d$-graph with no clique of size $r+1$. Then $k(H) \le {V(\frac{(d - 1)n}{r}, d-1)}^{\frac{r}{d-1}}$. In particular when $d-1 \mid r$ and $r \mid (d - 1)n$, the unique extremal case is the $\frac{r}{d-1}$-partite $d$-graph on $n$ vertices where hyperedges are all $d$-tuples which intersect each part in at most $d - 1$ vertices.
\end{conjecture}

Observe that \cref{thm:large-r-num} proves the conjecture for $r \ge (1-\frac{1}{d})n$. Let us make some comments regarding \cref{conj:turan} and how it compares with the usual Tur{\'a}n problem for hypergraphs. The Tur{\'a}n problem asks to determine the maximum number of hyperedges in a $d$-graph with no clique of size $r+1$. This is notoriously open even for $d = r = 3$. One explanation for the intractability of this problem is that unlike the case of graphs, there are exponentially many extremal examples for hypergraphs (see~\cite{MR685045}). In our case however we conjecture that there is a unique extremal example which might mean that the problem is easier. Moreover, for our application we do not need the full generality of the conjecture. In particular, it is sufficient for us to determine the case $d = 3$ and $r = \epsilon n$ for $\epsilon > 0$. Interestingly for some regime of these parameters the Tur{\'a}n number is known and has been rediscovered several times (see~\cite{MR927456,DBLP:journals/combinatorica/ChvatalM92,DBLP:journals/combinatorica/ThomasseY07}).

\begin{theorem}\label{thm:every-r}
Assuming \cref{conj:turan} holds, $u(n, r, d) \le {V(\frac{(d - 1)n}{r}, d-1)}^{\frac{r}{d-1}}$. In particular for every $\epsilon > 0$, $u(n, \epsilon n, 3) \le {(\binom{2/\epsilon}{2}+2/\epsilon+1)}^{\epsilon n/2}$.
\end{theorem}

\section{Depth-3 Circuits}

In this section we give applications of the $\un_2$ and $\un_3$ dimension to depth-3 circuits.

\subsection{Projections}

A \emph{projection} in $\Q^n$ is an affine space given by equations of the form $x = 0$, $x = 1$, $x = y$ or $x = 1 - y$. Given $S \subseteq \Q^n$, we denote by $\prj(S)$ the dimension of the largest projection contained in $S$. We define $\af(S)$ to be the dimension of the largest affine space contained in $S$. Note that $\af(S) \ge \prj(S)$ since a projection is a particular type of affine space. We will show that the converse is also true when $S$ is the set of satisfying assignments of a $2$-CNF.

A projection of dimension $d$ in a variable set $X$ can be represented as a sequence of $2(d+1)$ sets $(A_0, B_0, A_1, B_1, \ldots, A_d, B_d)$, where $\bigcup\limits_{i=0}^d A_i \cup B_i = X$ and for every $i \ge 1$ $A_i \cup B_i$ is non-empty. $A_0$ contains variables that are set to $0$, $B_0$ contains those set to $1$, and for $i \ge 1$ the variables from $A_i$ are equal to each other and the variables from $B_i$ are equal to their negations.

For a Boolean function $f$, we write $\prj(f)$ to denote $\prj(f^{-1}(1))$.

\begin{lemma}[Paturi, Saks, and Zane~\cite{DBLP:journals/cc/PaturiSZ00}]\label{lm:psz}
Let $S = \sat(\phi)$ for a $2$-CNF formula $\phi$. Then $\prj(S) = \vc(S)$.
\end{lemma}

Thus, by the Sauer--Shelah lemma, if $\prj(S) \le d$ for such $S$, then $\card{S} \le \sum_{i = 0}^d \binom{n}{i}$. We improve this bound by
showing that as far as 2-CNFs are concerned, $\vc$ and $\un_2$ dimensions are the same.

\begin{lemma}\label{lm:vdu2}
Let $S = \sat(\phi)$ for a $2$-CNF formula $\phi$. Then $\vc(S) = \un_2(S)$.
\end{lemma}
\begin{proof}
$\vc(S) \le \un_2(S)$ follows from the definition.
It remains to show that $\vc(S) \ge \un_2(S)$. Let $\phi$ be the $2$-CNF formula in a variable set $X$ with $S = \sat(\phi)$. Recall the \textit{implication digraph} of $\phi$, $D(F)$, which is constructed as follows. For every literal $u$ there is a vertex. For every clause $u \vee v$ we have two edges, $\overline{u} \rightarrow v$ and $\overline{v} \rightarrow u$. Every unit clause $v$ gives the edge $\overline{v} \rightarrow v$. We say that literal $u$ \emph{implies} literal $v$ if there is a directed path from $u$ to $v$.
Let $I \subseteq X$ be $2$-universal for $S$. It follows that for any $x, y \in I$ no literal on $x$ implies a literal on $y$, since otherwise setting a value of one forces the value of the other, contradicting $2$-universality. We claim that any assignment to the variables in $I$ can be extended to a full assignment satisfying $\phi$. Let $\alpha$ be any satisfying assignment of $\phi$. We follow the argument of~\cite{DBLP:journals/cc/PaturiSZ00}. There are different types of literals:

\begin{itemize}
\item Literals that imply some literal in $I$: we set such literals to $0$.
\item Those that are implied by some literal in $I$: we set these to $1$.
\item Those that are in the same strongly connected component with some literal in $I$: we set these as the one in $I$.
\item All other literals: we set these according to $\alpha$.
\end{itemize}

It is easy to see that it defines a satisfying assignment.

\end{proof}

We need the following lemma.

\begin{lemma}
Let $S \subseteq \Q^n$ be a $d$-dimensional affine space. Then $\vc(S) = d$.
\end{lemma}
\begin{proof}

Since $S$ is $d$-dimensional, there exists a full rank matrix $M \in \Q^{d\times n}$ and $c \in \Q^n$, such that $S = \{xM + c : x \in \Q^d\}$. Since $M$ is full rank, there exists a set $L \subseteq [n]$ of $d$ linearly independent columns. Let $M_L$ be the restriction of $M$ to the columns in $L$. It follows that $M_L$ is a full-rank $d\times d$ matrix. Observe that the projection of $S$ on $L$ is given by $\{xM_L + c : x \in \Q^d\}$. Since $M_L$ is full-rank, this equals to the whole $d$-dimensional space. Therefore, $L$ is shattered by $S$, and we are done.

\end{proof}

Combining these lemmas we get the following.

\begin{theorem}\label{lm:2sat_af_eq_u2}
  Let $S = \sat(\phi)$ for a $2$-CNF formula $\phi$. Then
  \[
    \af(S) = \vc(S) = \prj(S) = \un_2(S).
  \]
\end{theorem}

The following lemma directly follows from \cref{thm:utwo} and \cref{lm:vdu2}.

\begin{lemma}\label{lm:2sat_u2_size_bound}
  Let $S = \sat(\phi)$ for a $2$-CNF formula $\phi$ in $n$ variables. Then 
  \[\card{S} \le {\left(\frac{n}{\prj(S)} + 1\right)}^{\prj(S)}.\]
\end{lemma}

Although projections can be used to prove lower bounds on $\Sigma_3^3$ circuits, their application is quite limited. As noted in~\cite{DBLP:journals/cc/PaturiSZ00}, low-density parity-check codes suggested by Gallager~\cite{DBLP:journals/tit/Gallager62} can be used to give an example of such a limitation. Let $H$ be a parity-check matrix of such a code. It contains at most $4$ ones in each row. Therefore, the system $Hx=0$ can be represented as a $4$-CNF formula. This code has exponentially many codewords and only contains projections of constant size since its distance is linear.

This construction can easily be extended to $3$-CNF formulas. Consider any line of the linear system $Hx=0$ containing exactly four variables. Without loss of generality, we may assume that it depends on variables $x_1, x_2, x_3, x_4$, i.e., it is $x_1+x_2+x_3+x_4=0$. We replace it with two new lines $y+x_3+x_4=0$ and $y = x_1+x_2$, where $y$ is a fresh extension variable. We do this replacement for every line of $Hx=0$ with four variables. After this transformation, the new system can be represented as a $3$-CNF formula. Since the new extension variables are uniquely determined by the original variables, the new code has the same number of codewords and its distance is at least the distance of the original code. Hence, it can only contain a projection of at most constant size.

However, we show that a $3$-CNF with sufficiently many satisfying assignments accepts a projection of linear dimension.

\begin{lemma}\label{lm:hittingset}
Let $\phi$ be a $k$-CNF formula in a variable set $X$ and let $S = \sat(\phi)$. Assume that $I \subseteq X$ is $k$-universal for $S$. Then $X \setminus I$ is a hitting set for $\phi$, i.e., every clause $C \in \phi$ intersects $X \setminus I$.
\end{lemma}
\begin{proof}
Assume that this is not the case and a clause $C$ is entirely contained in $I$. Assume without loss of generality that $C = x_1 \vee \ldots \vee x_k$. By $k$-universality of $I$, $S$ contains an assignment which sets all these variables to 0. However, this assignment falsifies $C$ and hence cannot be in $S$, which is a contradiction.
\end{proof}

\begin{lemma}\label{lm:3cnf-proj}
There exists a universal constant $\delta > 0$ such that if $\phi$ is a $3$-CNF in $n$ variables accepting at least $7^{n/3} \simeq 2^{0.936 n}$ assignments, then $\prj(\phi) \ge \delta n$. Assuming \cref{conj:turan} holds, the statement holds (although for a different $\delta > 0$) when $\phi$ has at least $2^{0.707 n}$ satisfying assignments.
\end{lemma}

\begin{proof}
Let $S = \sat(\phi)$. Since $\card{S} \ge 7^{n/3}$, \cref{thm:large-r-num} implies that $\un_3(S) \ge 2n/3$. By \cref{lm:hittingset}, $\phi$ has a hitting set $J$ of size at most $n/3$. Under any restriction $\sigma$ of $J$, $\restr{\phi}{\sigma}$ is a $2$-CNF. Choose $\sigma$ such that $\restr{\phi}{\sigma}$ has at least $(7/2)^{n/3}$ satisfying assignment. \cref{lm:2sat_u2_size_bound} implies there exists $\delta > 0$ such that $\prj(\phi) \ge \prj(\restr{\phi}{\sigma}) \ge \delta n$.

Under \cref{conj:turan}, $u(n, 0.296 n, 3) \le 2^{0.706n}$. Thus, if $\phi$ has at least $2^{0.707 n}$ satisfying assignments, then $\phi$ has a hitting set of size at most $0.704 n$. Now we can apply the same argument as above.

\end{proof}

\subsection{Affine dispersers}

Recall that an affine disperser for dimension $d$ is a function, which is not constant under any affine space of dimension $d$.




\begin{theorem}\label{thm:sigma32_prj}
  Let $f : \Q^n \rightarrow \Q$ be an affine disperser for dimension $d+1$. Then
  \[
    s_3^2(f) \ge \frac{\card{f^{-1}(1)}}{{\left(\frac{n}{d} + 1\right)}^{d}}.
  \]
\end{theorem}
\begin{proof}
  Suppose that $f = \bigvee\limits_{i=1}^m \phi_i$, where $\phi_i$ are $2$-CNF formulas. It is clear that $\af(\phi_i) \le \af(f) \le d$.
  Let $S_i = \sat(\phi_i)$. Since $\af(\phi_i) = \prj(\phi_i)$ by \cref{lm:2sat_af_eq_u2},
  \cref{lm:2sat_u2_size_bound} implies that
  \[
    \card{S_i} \le {\left(\frac{n}{d} + 1\right)}^{d}.
  \]

  Hence,
  \[
    s_3^2(f) \ge m \ge \frac{\card{f^{-1}(1)}}{{\left(\frac{n}{d} + 1\right)}^{d}}.
  \]
\end{proof}

\begin{theorem}
  Let $f = \Q^n \rightarrow \Q$ be an affine disperser for dimension $d = o(n)$ with $\card{f^{-1}(1)} \ge 2^{n - o(n)}$. Then
  \[
    s_3^3(f) \ge 2^{0.064n - o(n)}.
  \]

  Furthermore, assuming \cref{conj:turan} holds,
  \[
    s_3^3(f) \ge 2^{0.293n - o(n)}.
  \]
\end{theorem}{}

\begin{proof}

We apply \cref{lm:3cnf-proj} and follow the same proof as \cref{thm:sigma32_prj}.

\end{proof}

%



\subsection{Degree-\texorpdfstring{$2$}{2} polynomials}

We now give a lower bound for all degree-2 polynomials over $\bbF_2$.

\begin{lemma}\label{lm:deg2_poly_ind_set}
  Let $p$ be a degree-$2$ polynomial over $\bbF_2$ in $n$ variables and $I$ a set of variables of $p$, no two of which produce a monomial of $p$. Then
  \[
    s_3^t(p) \ge 2^{\frac{\card{I}}{2t}}.
  \]
\end{lemma}
\begin{proof}
  We randomly assign the values of the variables not belonging to $I$ and denote the resulting polynomial as $q$. It is clear that $s_3^t(p) \ge s_3^t(q)$.

  By construction of $I$, $q$ is an affine function in at most $\card{I}$ variables.
  If a variable $z \in I$ does not appear in any degree-$2$ monomial of $p$, then $z$ always appears in $q$. Otherwise, $z$ appears in $q$ with probability $\frac{1}{2}$.
  It follows that the expected number of the variables that appear in $q$ is at least $\frac{\card{I}}{2}$.
  Therefore, there exists an assignment, such that $q$ is an affine function in at least $\frac{\card{I}}{2}$ variables.

  The parity function in $n$ variables requires $\Sigma_3^t$ circuit of size at least $2^{\frac{n}{t}}$~\cite{DBLP:journals/cc/PaturiSZ00}.
  Thus, $s_3^t(q) \ge 2^{\frac{\card{I}}{2t}}$.
\end{proof}

\begin{lemma}
  Let $p$ be a degree-$2$ polynomial over $\bbF_2$ in $n$ variables. Then
  \[
    s_3^2(f) \ge 2^{n/10}.
  \]
\end{lemma}
\begin{proof}
  Let $I$ be the largest set that satisfies the assumptions of \cref{lm:deg2_poly_ind_set}.

  If $\card{I} \ge \beta n$, then
  \[
    s_3^2(p) \ge 2^{\frac{\beta}{4} n}.
  \]

  Otherwise, there are no large projections that make $p$ constant.
  Consider a projection $(A_0, B_0, A_1, B_1, \ldots, A_d, B_d)$ of dimension $d$, which makes $p$ constant. Let $J$ be the set of all the variables contained in some $(A_i, B_i)$ for $i \ge 1$ with $\card{A_i \cup B_i} = 1$. $J$ satisfies the conditions for \cref{lm:deg2_poly_ind_set} since otherwise $p$ under the projection would have a monomial with a non-zero coefficient. Thus, $\card{J} < \beta n$ and every other part $(A_i, B_i)$ contains at least two variables.

  Therefore, we have the following:
  \[
    n = \sum\limits_{i=0}^d \card{A_i \cup B_i} \geq \card{J} + (d - \card{J}) 2 > 2d - \beta n.
  \]

  It gives an upper bound on $d$:
  \[
    d < \frac{1+\beta}{2} n.
  \]

  This, \cref{thm:sigma32_prj}, and a well-known fact that a degree-$2$ polynomial over $\bbF_2$ is $1$ on at least $2^{n-2}$ inputs implies that
  \[
    s_3^2(p) \ge 2^{\left( 1 - \frac{1+\beta}{2} \log(\frac{2}{1+\beta} + 1) \right) n - o(n)}.
  \]

  By choosing $\beta \approx 0.4$, we conclude that $s_3^2(p) \ge 2^{n/10}$.
\end{proof}

\section{The inner product and \texorpdfstring{$2$}{2}-CNF formulas}

It is more convenient for us to consider the negation of the inner product function $\IP$ on $k$ pairs of variables.
\begin{equation*}
    \IP(x_1, \ldots, x_k, y_1, \ldots, y_k) \coloneqq
    \begin{cases}
      1, & \text{if } \sum\limits_{i=1}^k x_i y_i \pmod 2 = 0 \\
      0, & \text{otherwise}.
    \end{cases}
\end{equation*}

\cite{DBLP:conf/innovations/GolovnevKW21} studied the following properties of Boolean circuits.
For an integer $k \ge 2$, $\alpha(k)$ is the infimum of all values $\alpha$ such that any circuit of size $s$ can be rewritten as an $\mathrm{OR}_{2^{\alpha s}} \circ \mathrm{AND} \circ \mathrm{OR}_k$ circuit.
The exact value is only known for $\alpha(2)$, and they showed that $\alpha(3) \le \frac{\log_2 3}{4}$.
The inner product function is a natural candidate for a hard function for $\Sigma_3^3$.

\begin{lemma}[\cite{DBLP:conf/innovations/GolovnevKW21}]
  \leavevmode%
  \begin{enumerate}
    \item $2^{\frac{k}{2}} \le s_3^2(\IP) \le 2^{k-o(k)}$.
    \item $2^{\frac{k}{3}} \le s_3^3(\IP) \le 3^{\frac{k}{2}}$.
  \end{enumerate}
\end{lemma}

Both lower bounds are obtained via a simple reduction to the parity function and the fact that $s_3^t(\oplus_n) \ge 2^{\frac{n}{t}}$~\cite{DBLP:journals/cc/PaturiSZ00}.
If the upper bound for $s_3^3(\IP)$ in the lemma is tight, then $\alpha(3) = \frac{\log_2 3}{4}$. However, the correct bound is not known even for $s_3^2(\IP)$.

We say that a CNF formula $\phi$ is \emph{consistent} with the inner product if $\phi^{-1}(1) \subseteq \IP^{-1}(1)$. We denote this as $\phi \leq \IP$.
The $2$-universality can be used to prove that every $2$-CNF formula consistent with the inner product has at most $3^{k}$ satisfying assignments. However, if applied directly, this only gives $2^{0.40k}$ lower bound, which is worse than the reduction to the parity function.

\begin{theorem}
  Let $\phi$ be a $2$-CNF formula consistent with the inner product on $k$ pairs of variables. Then
  \[
    \card{\sat(\phi)} \leq 3^k.
  \]
\end{theorem}
\begin{proof}
  Let $S = \sat(\phi)$.

  It is well-known that $\IP$ is a $k$-affine disperser (see, e.g.,~\cite{DBLP:conf/innovations/CohenS16}). Thus, $\af(S) \le k$. By \cref{lm:2sat_af_eq_u2}, $\un_2(S) = \af(S) \le k$.
  \Cref{lm:2sat_u2_size_bound} implies that
  \[
    \card{S} \leq {\left(\frac{2k}{k}+1\right)}^k = 3^k.
  \]
\end{proof}

What is more, the $2$-CNF formula that has this many satisfying assignments is unique.

Consider a $2$-CNF formula $\phi$ such that $\phi \leq IP$ and $\phi$ has the maximal possible number of satisfying assignments.
We will prove that there is only one $2$-CNF formula that has this many satisfying assignments:
\begin{equation*}
  \bigwedge\limits_{i=1}^k \left(\neg x_i \lor \neg y_i\right).
\end{equation*}

A \emph{transitive closure} $\tc(\phi)$ of a CNF formula $\phi$ is an equivalent CNF formula that contains all clauses that can be derived from $\phi$.

\begin{fact}
  Let $\phi$ be a $2$-CNF formula.
  $\phi \nvDash (z = \alpha)$ for every variable $z$ of $\phi$ and every $a \in \Q$ if and only if every clause of $\tc(\phi)$ has width $2$.
\end{fact}
%

Without loss of generality we can assume that $\phi = \tc(\phi)$.

We will first prove a general property of $2$-CNF formulas consistent with the inner product.

\begin{lemma}\label{lem:forbidden_assignments}
  Consider a 2-CNF formula $\phi$ that is consistent with inner product on $k$ pairs of variables.
  Let $J \subseteq \left[k\right], J \neq \varnothing$.
  Suppose that $\phi$ has two satisfying assignments $\sigma$ and $\tau$ such that:
  \begin{itemize}
    \item $\sigma(x_i) = 0$ and $\sigma(y_i) = 1$ for $i \in J$.
    \item $\tau(x_i) = 1$ and $\tau(y_i) = 0$ for $i \in J$.
    \item $\sigma(x_i) = \tau(x_i)$ and $\sigma(y_i) = \tau(y_i)$ for $i \in \left[k\right] \setminus J$.
  \end{itemize}

  Then at least one of the following holds:
  \begin{enumerate}
    \item\label{case:zero_edge} $\phi \vDash (x_i y_i = 0)$ for some $i \in J$.
    \item\label{case:equality_on_xs} $\phi \vDash (x_i = x_j)$ for some distinct $i, j \in J$.
  \end{enumerate}
\end{lemma}
\begin{proof}
  Let $\rho$ be the common part of $\sigma$ and $\tau$ (i.e., $\rho$ is a partial assignment to $x_i$ and $y_i$, where $i \in \left[k\right] \setminus J$). Consider the restricted formula $\psi = \restr{\phi}{\rho}$.

  For every variable $z$ of $\psi$ we have $\sigma(z) = \neg\tau(z)$. Therefore, $\psi$ cannot have clauses of length $1$.

  Fix $i \in J$ and consider the assignment $\sigma'$ that coincides with $\sigma$ except for the value of $x_i$: $\sigma'(x_i) = 1$. This assignment cannot satisfy $\psi$ since we flipped the value of only one monomial $x_i y_i$ without changing anything else. Every clause that is falsified by $\sigma'$ must contain $\neg x_i$ and have length $2$. Thus, it can be either of these:
  \begin{enumerate}
    \item $\neg x_i \lor x_j$ for some $j \in J, j \neq i$.
    \item $\neg x_i \lor \neg y_j$ for some $j \in J$.
  \end{enumerate}

  If there is a clause of the second type with $j=i$, then $\phi \vDash (x_i y_i = 0)$, and it concludes the proof.

  If it is not the case, we show that there must be a clause of the first type. Assume the opposite: there is a set $A \subseteq J \setminus \{i\}$ such that $\neg x_i \lor \neg y_j$, where $j \in A$, are the only clauses that are falsified by $\sigma'$.
  Define another assignment $\sigma''$ as follows: $\sigma''(y_j) = 0$ if $j \in A$ and $\sigma''(z) = \sigma'(z)$ otherwise.
  We show that the assignment $\sigma''$ satisfies $\psi$. Firstly, note that by construction it satisfies all the clauses that are falsified by $\sigma'$.

  Suppose that a clause $y_j \lor \ell$ is not satisfied by $\sigma''$, where $\ell$ is a literal. Observe that for any $t$, $\ell \neq y_t$ since $\tau$ sets every $y_t$ to $0$ and $\tau$ satisfies $\psi$. Also, $\ell \neq \neg y_t$, where $t \in A$, since in this case $\sigma''(y_t) = 0$. Hence, $\sigma''(\ell) = \sigma'(\ell)$. We can resolve this clause with $\neg x_i \lor \neg y_j$ and get $\neg x_i \lor \ell$. Thus, $\neg x_i \lor \ell$ must also be unsatisfied by $\sigma''$, which is a contradiction to the fact that $\sigma''$ satisfies all the clauses falsified by $\sigma'$.

  On the other hand, $\sum\limits_{i \in J} \sigma'(x_i) \sigma'(y_i) = \sum\limits_{i \in J} \sigma''(x_i) \sigma''(y_i)$. Therefore, $\sigma''$ cannot be a satisfying assignment of $\psi$.

  Hence, for every $i \in J$ there exists $j \in J \setminus \{i\}$ such that $\neg x_i \lor x_j$ is a clause of $\psi$. The conjunction of these clauses implies $x_i = x_j$ for some distinct $i$ and $j$.
\end{proof}

We are now ready to prove the uniqueness of the extremal 2-CNF.

\begin{theorem}\label{thm:ip_max_assignments}
  Let $\phi$ be a $2$-CNF formula consistent with the inner product that has the maximum number of satisfying assignments, i.e., $\card{\sat(\phi)} = 3^k$. Then for every $i \in \left[k\right]$ it holds that $\phi \vDash (x_i y_i = 0)$. Therefore, $\phi$ is equivalent to $\bigwedge\limits_{i=1}^k (\neg x_i \lor \neg y_i)$.
\end{theorem}

\begin{proof}
  We prove the statement by induction on $k$. The base case $k=1$ is clear.

  For the inductive step, we use \cref{lem:forbidden_assignments}.

  First we show that it is enough to show that $\phi \vDash (x_i y_i = 0)$ for at least one $i \in [k]$.

  \begin{claim}\label{clm:one_implies_every}
    Suppose that there exists $i \in \left[k\right]$ such that $\phi \vDash (x_i y_i = 0)$.
    Then for \emph{every} $i \in \left[k\right]$ it holds that $\phi \vDash (x_i y_i = 0)$.
  \end{claim}
  \begin{claimproof}
    There are three ways of setting $x_i y_i$ to $0$. For every satisfying assignment $\sigma$ of $\phi$ we have one the following:
    \begin{itemize}
      \item $\sigma(x_i) = \sigma(y_i) = 0$.
      \item $\sigma(x_i) = 0$ and $\sigma(y_i) = 1$.
      \item $\sigma(x_i) = 1$ and $\sigma(y_i) = 0$.
    \end{itemize}

    Now choose a partial assignment $\rho$ of the variables $x_i$ and $y_i$, such that the number of satisfying assignments of $\phi$ that coincide with $\rho$ is maximal. Then $\card{\sat(\restr{\phi}{\rho})} \geq 3^{k-1}$ and we can apply the inductive hypothesis.
  \end{claimproof}

  If $\phi$ has satisfying assignments that satisfy the assumptions of \cref{lem:forbidden_assignments}, then either $\phi$ implies $x_i y_i = 0$ for some $i \in \left[k\right]$, and we can apply the claim above, or $\phi$ implies $x_i = x_j$ for some $i, j \in [k], i \neq j$.
  In the latter case, we show that $\phi$ has less than $3^k$ satisfying assignments. Under this assumption, every satisfying assignment of $\phi$ satisfies
  \[
    x_i(y_i+y_j) + \sum_{s \in [k] \setminus \{i,j\}} x_s y_s = 0.
  \]

  Let $\rho$ be an arbitrary assignment to $x_i,y_i,y_j$. By \cref{thm:ip_max_assignments}, $\restr{\phi}{\rho}$ has at most $3^{k-2}$ satisfying assignments. Therefore, in total $\phi$ can have no more than $8\cdot 3^{k-2} < 3^k$ satisfying assignments.

  To conclude the proof, we show that if $\phi$ does not satisfy the assumptions of \cref{lem:forbidden_assignments}, then the number of satisfying assignments $\card{\sat(\phi)}$ is strictly less than $3^k$. We want to count the number of satisfying assignments in this case.
  By the definition of $\IP$, only an even number of monomials $x_i y_i$ can be set to $1$.
  Thus, for every satisfying assignment $\sigma$ of $\phi$ there exists a set $I \subseteq [k]$ of even size such that $\sigma(x_i) = \sigma(y_i) = 1$ for $i \in I$ and $\sigma(x_i) \sigma(y_i) = 0$ for $i \not\in I$.
  Let $J \subseteq [k] \setminus I$ be the set of all the indices $i$ satisfying $\sigma(x_i) = \sigma(y_i) = 0$. Since we assume that we cannot apply \cref{lem:forbidden_assignments}, there can be at most one satisfying assignment of $\phi$ for every choice of $I$ and $J$.

  Thus, the total number of satisfying assignments of $\phi$ with fixed $I$ is as most $2^{k-\card{I}}$.
  It follows that
  \begin{equation*}
    \card{\sat(\phi)} \leq \sum\limits_{\substack{I \subseteq [k] \\ \card{I} \text{ is even}}} 2^{k-\card{I}} = \sum\limits_{s=0}^{k/2} \binom{k}{2s} 2^{k-2s} = \frac{1}{2}\left( 3^k + 1 \right) < 3^k.
  \end{equation*}
\end{proof}

\section{Conclusion}

The most immediate problem which remains open is determining the exact $\Sigma_3^2$ and eventually $\Sigma_3^3$ complexity of $\IP$. It would be particularly pleasant if this is resolved using a stability argument extending our result on the uniqueness of 2-CNFs consistent with $\IP$ with the maximum number of satisfying assignments.

More generally collecting new combinatorial insights on the set of satisfying assignments of $k$-CNFs seems necessary to make progress towards $\Sigma_3^k$ lower bounds (and $k$-SAT which we did not cover in this paper).

Our work immediately raises the following natural question. Can we obtain better Sauer--Shelah lemmas for $k$-CNFs, i.e., given $d, k, n$ what is the largest size of a set $S \subseteq \Q^n$ with $\vc(S) = d$ which is the set of satisfying assignments of a $k$-CNF formula? We showed that for $k = 2$ this bound is ${\left(\frac{n}{d}+1\right)}^d$.

\bibliography{ref}

\end{document}